\documentclass[a4paper,runningheads,cleveref,thm-restate]{lipics-v2021}
\hideLIPIcs

\usepackage[T1]{fontenc}
\usepackage{lmodern}

\nolinenumbers
\usepackage{microtype} 
\usepackage{mathtools}
\usepackage{xspace}
\usepackage{nicefrac}
\usepackage{subcaption}

\usepackage[T1]{fontenc}
\usepackage{graphicx}
\usepackage{xcolor}

\newcommand{\Vis}[1]{\ensuremath{\mathrm{Vis}(#1)}}
\newcommand{\VisT}[1]{\ensuremath{\mathrm{Vis}^T(#1)}}
\newcommand{\VisHug}[2]{\ensuremath{\mathrm{Vis}^#1(#2)}}

\renewcommand{\int}{\ensuremath{\mathrm{int}}\xspace}

\newcommand{\hset}{solution set\xspace}

\newcommand{\invertible}{flexible\xspace}
\newcommand{\noninvertible}{inflexible\xspace}

\newcommand{\pseudoconvex}{pseudo-convex\xspace}
\newcommand{\etal}{{et al.}\xspace}
\newcommand{\mypar}[1]{\medskip\noindent{\bfseries\boldmath#1}}

\newcommand{\hug}{hugging cycle\xspace}
\newcommand{\HUG}{Hugging Cycle\xspace}

\crefname{lemma}{Lemma}{Lemmas}
\crefname{figure}{Figure}{Figures}

\title{Augmenting Plane Straight-Line Graphs to Meet Parity Constraints}

\author{Aleksander Bjørn Grodt Christiansen}{Technical University of Denmark, Denmark }{abgch@dtu.dk}{https://orcid.org/0000-0002-9486-9115}{Supported by the VILLUM Foundation grant 37507 ``Efficient Recomputations for Changeful Problems''}
\author{Linda Kleist}{Universität Potsdam, Potsdam, Germany}{linda.kleist.1@uni-potsdam.de}{https://orcid.org/0000-0002-3786-916X}{}
\author{Irene Parada}{Universitat Politècnica de Catalunya, Spain }{irene.parada@upc.edu}{https://orcid.org/0000-0003-3147-0083}{Serra Húnter fellow and acknowledges the support of Independent Research Fund Denmark grant 2020-2023 (9131-00044B) ``Dynamic Network Analysis'', the Margarita Salas Fellowship funded by the Ministry of Universities of Spain and the EU (NextGenerationEU), and grants PID2019-104129GB-I00 and PID2023-150725NB-I00 funded by MICIU/AEI/10.13039/
	501100011033.}
\author{Eva Rotenberg}{Technical University of Denmark, Denmark}{erot@dtu.dk}{https://orcid.org/0000-0001-5853-7909}{Partially supported by the VILLUM Foundation grant 37507 ``Efficient Recomputations for Changeful Problems'' and the Independent Research Fund Denmark grant 2020-2023 (9131-00044B) ``Dynamic Network Analysis''.}

\authorrunning{A.B.G. Christiansen, L. Kleist, I. Parada, E. Rotenberg}
\relatedversion{A version of this paper was presented at the 50th International Workshop on Graph-Theoretic Concepts in Computer Science (WG 2024).}

\begin{document}

\maketitle
\begin{abstract}
	Given a plane geometric graph $G$ on $n$ vertices, we want to augment it so that given parity constraints of the vertex degrees are met.
	In other words, given a subset $R$ of the vertices, we are interested in a plane geometric 
	supergraph $G'$ such that exactly the vertices of $R$ have odd degree in $G'\setminus G$.
	We show that the question whether such a supergraph exists can be decided in 
	polynomial time 
	for two interesting cases. 
	First, when the vertices are in convex position, we present a linear-time algorithm. 
	Building on this insight, we solve the case when $G$ is a plane geometric path in $O(n \log n)$ time. 
	This solves an open problem posed by Catana, Olaverri, Tejel, and Urrutia (Appl.\ Math.\ Comput. 2020). 	
\end{abstract}

\section{Introduction}

A fundamental class of problems in graph drawing concerns augmenting existing embedded graphs, such that the resulting graph has some desired properties. In this paper, we approach the natural question of augmenting graphs to meet degree constraints. Given a geometric graph, that is, a graph together with a planar straight-line embedding, the problem is to augment the graph with straight-line edges such that constraints are met concerning the degrees of vertices in the resulting plane geometric graph. Even in the simplest version of this problem, where the degree constraints are modulo two, even or odd degree, the problem is NP-hard for general graphs~\cite{Catanaetal}, and conjectured to be NP-hard even for trees~\cite{Catanaetal}. In this paper, we resolve the question about the tractability of parity constraint augmentation for geometric paths, showing it to be near-linear time solvable. 

Our solution uses a connection between the parity constraint satisfaction problem and the problem of finding a compatible geometric spanning tree of a geometric graph. We use this connection to characterize a 
class of geometric paths which always admit an augmentation (except for obvious negative cases). The remaining paths share enough structural properties with paths on convex point sets, that we can use the same approach for these as for graphs on points in convex position. For such graphs, we give a surprisingly simple solution that computes the answer via a linear bottom-up traversal of the (arbitrarily rooted) weak dual graph.

\mypar{The Planar Parity Constraint Satisfaction Problem.}
A \emph{geometric graph} $G=(V,E)$ is a graph drawn in the plane 
such that its vertex set $V$ is a point set in general position (no three points are collinear)
and its edge set $E$ is a set of straight-line segments between those points.
A geometric graph is \emph{plane} if no two of its edges cross and it is \emph{convex} if its vertices are in convex position (regardless of its edges). 
The \emph{visibility graph} of a plane geometric graph $G$, denoted by \Vis{G}, is a geometric graph that has $V$ as its vertex set and two vertices $u$ and $v$ share an edge in \Vis{G} if and only if $uv \notin E$ and $uv$ does not cross any edge in $E$ (so $G$ can be augmented by $uv$). 
Note that \Vis{G} might have many edges and can be far from being plane; see \Cref{fig:IntroA} for an example. 

\begin{figure}[hbt]
	\centering
	\begin{subfigure}{.45\textwidth}
		\centering
		\includegraphics[page=1]{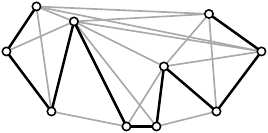}
		\caption{}
		\label{fig:IntroA}
	\end{subfigure}\hfill
	\begin{subfigure}{.45\textwidth}
		\centering
		\includegraphics[page=2]{figures_final}
		\caption{}
		\label{fig:IntroB}
	\end{subfigure}
	\caption{(a) A plane geometric graph $G$ (in black) and its visibility graph \Vis{G} (in gray). (b) A \hset for $G$ and the four unhappy vertices (red squares).}
	\label{fig:Intro}
\end{figure}

We are interested in plane subgraphs $G^+$ of \Vis{G} such that $G\cup G^+$ satisfies given parity constraints for the vertex degrees. 
These constraints can be interpreted as a set of \emph{unhappy} vertices $R$ that would like to change the parity of their degree in $G$. 
We refer to vertices that are not unhappy as \emph{happy}. 
A set of edges $H$ in \Vis{G} is called a \emph{\hset} for $(G,R)$ if $H$ is crossing-free and the vertices that have odd degree in $H$ are exactly the vertices in $R$; see \Cref{fig:IntroB}. 
Note that 
by the handshaking lemma, 
a \hset can only exist if $|R|$ is even.

\mypar{Our contribution.}
In Section~\ref{sec:convex}, we present a linear-time algorithm that, given a convex plane geometric graph $G = (V,E)$ and any set $R$ of unhappy vertices, outputs whether there exists a \hset.
Building on these insights, in Section~\ref{sec:paths}, we present a conceptually elegant $O(|V|\log |V|)$-time algorithm that decides the existence of \hset{s} 
for plane geometric paths, solving an open problem by Catana~\etal\cite{Catanaetal}. We also characterize those paths that admit a positive answer for any even subset $R$ of unhappy vertices.

\mypar{Related work.}
In graph augmentation problems, one is typically interested in adding (few) new edges to a given graph in order to achieve some desired property; for instance, 2-connectivity, a reduced diameter, to obtain an Eulerian graph, or meet other desired vertex degrees. 
These problems arise in many applications like transportation and telecommunication networks and constitute an active research topic in graph theory and algorithms. 
Augmenting graphs to meet connectivity constraints is well-studied in the general case~\cite{BenczurK00,CenLP22,Tarjan,Frank92,Plesnik1976} as well as in the planar case~\cite{KantBodl1991,RutterWolff2012}, and in the plane geometric case~\cite{ABELLANAS2008,AKITAYA2019,AlJubeh2011,Garcia2015,KRANAKIS2010,RutterWolff2012,TOTH2012}. 
Reducing the diameter has also received considerable attention~\cite{AdriaensG22,AlonGR00,ChungG84,FratiGGM15}.

{The study of augmentation to meet parity constraints and, in particular, structures involving the set of odd-degree vertices has a long history. 
\emph{$T$-joins} are subgraphs in which the set of vertices of odd degree is exactly $T$. 
The minimum-weight $T$-join problem can be efficiently solved and generalizes several problems including the Chinese postman problem~\cite{EdmondsJ73,ChinesePostman1}. 
The concept of odd-vertex pairing also appears in classic results~\cite{maxcut,nash-williams_1960}.
}

Augmentation to meet parity constraints has also been considered as a relaxed version of the more general problem of augmenting graphs to meet a certain degree sequence. 
For abstract graphs, Dabrowski~\etal\cite{DABROWSKI2016213} provide a polynomial-time algorithm to achieve parity constraints. 
For both the plane topological setting, where edges are not required to be straight lines, and the plane geometric setting, 
Catana~\etal\cite{Catanaetal} showed that deciding whether a general plane (geometric) graph can be augmented to meet a set of parity constraints is NP-complete; even in more restricted cases where the set of unhappy vertices $R$ is $V$ or all odd-degree vertices 
(i.e., the goal is to make the graph Eulerian).
For maximal outerplane graphs in the topological setting, 
they show how to decide in polynomial time if the parity constrains can be met in all of the cases above. 
They also show that the same holds if one can only augment by adding a matching. 
In contrast, deciding whether a plane geometric cycle 
can be augmented with a (perfect) geometric matching so that 
all vertices change parities 
is known to be NP-complete~\cite{PilzRS020}.

The geometric setting has also been studied in a slightly different light~\cite{Aichholzer2014,Alvarez2015}. Given a set of vertices along with a set of parity constraints for these vertices, Aichholzer~\etal\cite{Aichholzer2014} showed that one can always construct a plane spanning tree and a 2-connected outerplane graph satisfying all of the constraints. 
Garc\'ia~\etal~\cite{GARCIA2014} study the following question: Given a plane geometric tree $T$ on a point set $S$, find a $T$-compatible plane spanning tree on $S$ sharing the minimum number of edges with $T$. 
They show that this corresponds to finding a spanning forest of the visibility graph of $T$ with the minimum number of components.
For certain plane geometric paths, they prove that one can always find a spanning tree of the visibility graph.

\section{Geometric Graphs with a \HUG}
\label{sec:convex}

{For a plane geometric graph $G=(V,E)$, a \emph{\hug} $C$ is a plane geometric (simple) cycle 
with vertex set $V$ 
such that the edges of $G$ lie in the (closed) interior of $C$. 
Note that a plane geometric graph (or path) might not admit a \hug or might admit more than one.
We denote by 
$\VisHug{C}{G}$ the restriction of $\Vis{G}$ to the edges that lie in the (closed) interior of $C$.} 

An important concept towards our result is a \emph{convexly hugging cycle}; a hugging cycle $C$ of a geometric plane graph $G$ such that all bounded faces of $G\cup C$ are convex.
The main result in this section is the following theorem.  

\begin{theorem}
	\label{thm:HUG}
	Let $G=(V,E)$ be a plane geometric graph and let $C$ be a convexly \hug{} of $G$, and let $R\subset V$ be a set of vertices. Then, there is an $O(|V|)$-time algorithm for deciding whether $\VisHug{C}{G}$ contains a subgraph where exactly the vertices of $R$ have odd degree.
\end{theorem}

{If $G$ is a convex plane geometric graph, then the edges on the convex hull constitute a (in fact, unique) \hug $C$ and $\VisHug{C}{G}=\Vis{G}$, which, furthermore, is convexly hugging. We thus directly obtain the following corollary:}

\begin{corollary}
	\label{cor:convexPos}
	Let $G=(V,E)$ be a convex plane geometric graph, and let $R\subseteq V$ be the set of unhappy vertices. There exists an $O(|V|)$-time algorithm to decide whether $(G,R)$ admits a solution, 
	assuming the convex hull is part of the input.
\end{corollary}

We remark that
if $G$ is a connected convex plane geometric graph, the convex hull can be computed in linear time~\cite{Melkman87}. 
In contrast, for disconnected convex plane geometric graphs, computing  
the \hug is lower bounded by $\Omega(|V|\log (|V|))$ 
in the decision tree model of computing. 

Now we present the proof of \Cref{thm:HUG}.
Let $C$ denote the convexly \hug for~$G$. 
In a first step, we consider a geometric subgraph $G_f$ of $G$ induced by a bounded (convex) face $f$ of $G\cup C$, i.e., all edges are on the convex hull. To characterize  the existence of \hset{s} for these cases,
we distinguish three scenarios: 
(i) all edges are present, i.e., $G_f$ is a cycle (\Cref{lemma:convexCycle}); 
(ii) exactly one edge is missing, i.e., $G_f$ is a path (\Cref{lemma:convexPath}); and
(iii) at least two edges are missing, i.e., $G_f$ is a collection of paths and isolated vertices (\Cref{lemma:convexPathsCollection}).

\begin{restatable}{lemma}{convexCycle}
	\label{lemma:convexCycle}
	Let $C=(V,E)$ be a convex plane geometric cycle and $R\subseteq V$.  
	Then $(C,R)$ admits a \hset if and only if {$|R|=0$ or} $|R|$ is even and $V\setminus R$ contains two non-adjacent vertices.
\end{restatable}
\begin{proof}
	First, we show that the conditions are necessary. 
	Suppose there exists a \hset $H$. 
	As noted above, by the handshaking lemma, 
	$|R|$ must be even. 
		If $|V|\le 3$ there exists a \hset if and only if $|R|=0$. 
		If $|V| > 3$, note that for any plane subgraph $G^+$ of \Vis{G}, the weak dual graph $D(G\cup G^+)$ is a non-trivial tree; see \Cref{fig:cycleA} for an illustration. 
	Its (at least) two leaves certify two non-adjacent vertices which are not incident to any edge of $G^+$. 
	Applying this to $H$ proves that there must exist at least two happy vertices which are not consecutive along $G$. 
	\begin{figure}[htb]
		\centering
		\begin{subfigure}{.3\textwidth}
			\centering
			\includegraphics[page=3]{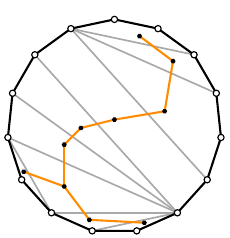}
			\caption{}
			\label{fig:cycleA}
		\end{subfigure}\hfill
		\begin{subfigure}{.3\textwidth}
			\centering
			\includegraphics[page=5]{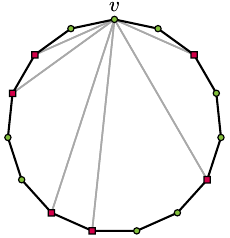}
			\caption{}
			\label{fig:cycleB}
		\end{subfigure}
		\hfill
		\begin{subfigure}{.3\textwidth}
			\centering
			\includegraphics[page=4]{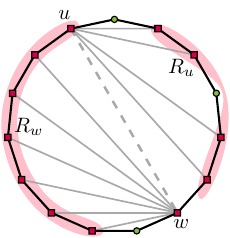}
			\caption{}
			\label{fig:cycleC}
		\end{subfigure}
		\caption{{Illustration for \cref{lemma:convexCycle} and its proof. (a) The weak dual graph $D(G\cup G^+)$ (in orange) which is a tree.  (b) A \hset consisiting of a star with center $v$ and vertices of $R$ as leaves.  (c) A \hset formed by a double star with centers $u$ and $w$ and vertices of $R$ as leaves (where the edge $uw$ might be removed if $|R_u|$ is even).}}
		\label{fig:cycle}
	\end{figure}
	
	It remains to show that the conditions are also sufficient {when $|V| > 3$}. 
	Assume that $|R|$ is even and there are at least two non-consecutive happy vertices. 
	We consider two cases: either $V\setminus R$  contains three consecutive happy vertices or there exist two non-adjacent happy vertices $u$ and $w$ 
	such that both paths in $G$ between $u$ and $w$ contain an unhappy vertex.  
	
	In the first case, let $v$ denote a vertex such that $v$ and its neighbors are happy vertices.
	Then, connecting $v$ with the vertices in $R$ defines the \hset  $H:=\{vr \mid r\in R\}$; see \Cref{fig:cycleB}  for an illustration. 
	It is easy to check that all degree constraints of $H$ are met: $v$ has even degree because $|R|$ is even and each vertex $r\in R$ has degree 1 in $H$. 
	
	In the second case, let $u$ and $w$ denote two non-adjacent unhappy vertices in $R$ such that their clockwise neighbor in $G$ is happy. 
	Introducing the artificial chord $uw$ splits $R\setminus\{u,w\}$ into two sets $R_u$ and $R_w$ containing the vertices in $R$ on a clockwise $uw$-path and $wu$-path, respectively. 
	We define $H:=\{ur| r\in R_u\}\cup \{wr| r\in R_w\}$. 
	If $|R_u|$ is odd (implying that $|R_w|=|R|-|R_u|-2$ is odd as well), we add $uw$ to $H$; see also \Cref{fig:cycleC}. 
	It is easy to check that all degree constraints are met and $H$ is a \hset.
\end{proof}

\begin{restatable}{lemma}{convexPath}
	\label{lemma:convexPath} 
	Let $G=(V,E)$ be a convex geometric path with $|V|\ge 3$ and 
	all edges on the boundary of the convex hull. 
	Given $R\subseteq V$, 
	there exists a \hset for $(G,R)$ if and only if $|R|$ is even and $V\setminus R$ contains an internal vertex of~$G$. 
\end{restatable}
\begin{proof}
	Let $s$ and $t$ denote the two endpoints of $G$.
	In comparison to \Cref{lemma:convexCycle}, there exists exactly one more variable of choice, namely to add or not to add the edge $st$ in $H$, see also \cref{fig:convexPath}.
	Let $C$ denote the cycle obtained from $G$ by adding $st$ and define $R'$ from $R$ by switching the happiness status of $s$ and $t$. Clearly, there exists a \hset $H$ for $(G,R)$ if and only if there exists a \hset for $(C,R')$ or for $(C,R)$; the edge $st$ will belong to $H$ in the first case and will not belong to $H$ in the second case. Note that $R$ and $R'$ have the same parity.
		Thus, by  \Cref{lemma:convexCycle}, there exists a \hset $H$ if and only if $|R|=|R'|$ is even and at least one of $V\setminus R$ and $V\setminus R'$ contains two non-adjacent vertices.

	\begin{figure}[htb]
		\centering
		\begin{subfigure}{.3\textwidth}
			\centering
			\includegraphics[page=25]{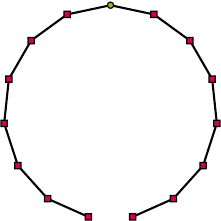}
			\caption{}
			\label{fig:convexPathvA}
		\end{subfigure}\hfill
		\begin{subfigure}{.3\textwidth}
			\centering
			\includegraphics[page=26]{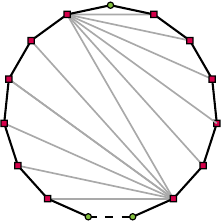}
			\caption{}
			\label{fig:convexPathB}
		\end{subfigure}\hfill
		\begin{subfigure}{.3\textwidth}
			\centering
			\includegraphics[page=27]{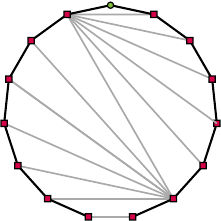}
			\caption{}
			\label{fig:convexPathC}
		\end{subfigure}
		
		\caption{{Illustration for \cref{lemma:convexPath} and its proof. (a) A convex geometric path~$G$ and a set $R$ of unhappy vertices (in red squares). (b) A \hset for the corresponding instance $(C,R')$. (c) A \hset for $G$.}}
		\label{fig:convexPath}
	\end{figure}
	
		If $V\setminus R$ or $V\setminus R'$ contains two non-adjacent vertices, then clearly at least one of them is an internal vertex $w$ of $G$.
		Now suppose that there is an internal happy vertex $w\in V\setminus R$ (and $|R|$ even).
		If $|V|=3$, either $|R|=0$ and hence $(C,R)$ has a \hset or $R=\{s,t\}$ and thus $|R'|=0$ and $(C,R')$ has a \hset.
		If $|V|>3$,  $w$ is not adjacent to $s$ or to $t$, say $s$.  
		If $s\in V\setminus R$, $s$ and $w$ are a certificate for an \hset of $(C,R)$ and if $s\in R$, $s$ and $w$ are a certificate for an \hset of $(C,R')$, implying that there is also a \hset for $(G,R)$.
\end{proof}

\begin{restatable}{lemma}{convexPathsCollection}
	\label{lemma:convexPathsCollection}
	Let $G=(V,E)$ be a collection of at least two geometric paths (could be singleton sets) on a convex point set with $|V|\geq 3$ and all edges on the boundary of the convex hull. 
	Given $R\subseteq V$, 
	there exists a \hset for $(G,R)$ if and only if $|R|$ is even.
\end{restatable}
\begin{proof}
	As before, $|R|$ even is clearly necessary. It remains to show that this is sufficient.
	Let $e_1$ and $e_2$ denote two edges of the convex hull that do not belong to $G$. If $|V|= 3$ and $|R|=2$, $H$ consists either of the edge between the vertices in $R$ (if this edge does not belong to $G$) or of $e_1$ and $e_2$. Hence, we assume that $|V|\geq 4$. In this case, we can choose a vertex $v_i$ from $e_i$ for $i=1,2$ such that $v_1$ and $v_2$ are not adjacent. 
	By adding or not adding $e_i$ to $H$, we can ensure that $v_i$ does not belong to $R$. 
	Then, \Cref{lemma:convexCycle} guarantees the existence of a \hset.
\end{proof}

Using \Cref{lemma:convexCycle,lemma:convexPath,lemma:convexPathsCollection}, we present a linear-time algorithm for any plane geometric graph with a convexly \hug. To this end, we use the concept of the weak dual.
{The \emph{weak dual} of a plane graph is the plane graph that has a
	vertex for each bounded face and an edge between two vertices if and only if the corresponding faces share an edge. 
	For a plane geometric graph $G=(V,E)$ 
	with a convexly \hug $C$, 
	the weak dual graph of $G\cup C$
	is a tree since it is the weak dual of a biconnected outerplanar graph. 
} 
For clarity, we sometimes refer to the vertices of the weak dual as the faces they correspond to.

\begin{proof}[Proof of \Cref{thm:HUG}]
	{
	Consider the convex weak dual graph $D(G)$ of $G \cup C$. 
	As remarked above, $G \cup C$ is a biconnected outerplanar geometric graph and thus, $D(G)$ is a tree. 
	We root $D(G)$ at a leaf $f_0$ and orient all edges towards $f_0$, making $D(G)$ an in-arborescence. See \Cref{fig:convexTreeA} for an illustration on a convex point set, and note that our proof holds in the more general case of plane geometric graphs with a convexly hugging cycle.} 	
	\begin{figure}[htb]
		\centering
		\begin{subfigure}{.3\textwidth}
			\centering
			\includegraphics[page=6]{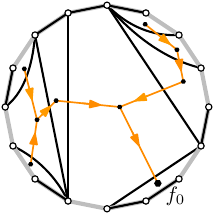}
			\caption{}
			\label{fig:convexTreeA}
		\end{subfigure}\hfil
		\begin{subfigure}{.3\textwidth}
			\centering
			\includegraphics[page=7]{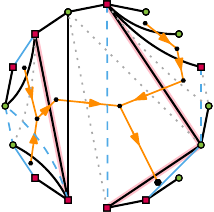}
			\caption{}
			\label{fig:convexTreeB}
		\end{subfigure}
		\caption{(a) A plane geometric graph $G$, a convexly \hug (gray), and their dual tree $D(G)$ rooted at $f_0$. (b) A \hset (dashed and solid blue). The highlighted edges (pink) are marked \invertible in the execution of the algorithm.}
		\label{fig:convexTree}
	\end{figure}

For every face $f$ that is not the root, we define the \emph{connector edge} $uv$ to be the edge of $G$ that $f$ shares with its parent in $D(G)$.  
At the root $f_0$, let the connector edge be an edge of $C$ incident to this face. 
Note that since $f_0$ is a leaf of $D(G)$ such an edge exists and it is not shared with the child face.

	The idea of the algorithm is to consider $D(G)$ bottom up, taking care at each step of a face $f$ of $G\cup C$ whose children have already been processed. 
	For a face~$f$ (different from the root) with connector edge $uv$, 
	we define $\Pi(f,p_u,p_v)$ to be the problem asking whether the geometric graph corresponding to the subtree of $D(G)$ rooted at $f$ has a solution that satisfies the given parity constraints, where the parity constraints at $u$ and $v$ are replaced by $p_u$ and $p_v$.     

	To solve $\Pi(f,p_u,p_v)$, we use bottom-up dynamic programming in the tree $D(G)$. Note that the graph $G_f$ restricted to face $f$ has all edges on its convex hull. 
	Hence, \Cref{lemma:convexCycle,lemma:convexPath,lemma:convexPathsCollection} characterize when there exists a \hset. 
	Note that these lemmas are constructive and provide a \hset in linear time when it exists. 
	Consider a child $f'$ of $f$ with connector edge $wx$; 
	we assume that problem $\Pi(f',p_w,p_x)$ has already been solved for all possible values of $p_w,p_x$. 
	We call $f'$ and its connector edge \emph{\noninvertible} if there is only one set of parities $p_w,p_x$ such that $\Pi(f',p_w,p_x)$ has a solution, otherwise we call them \emph{\invertible}. 
	The example in \Cref{fig:sketch-badcase-convex-tree} shows that it is necessary to do the book-keeping of \invertible faces; otherwise we may falsely conclude that there exists no \hset. 
	
	\begin{figure}[htb]
		\centering
		\includegraphics[page=8]{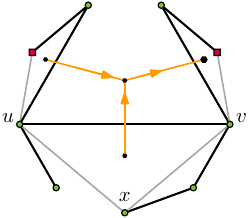}
		\caption{A convex plane geometric tree $T$ (black), a set $R$ (red), and the unique \hset (in gray). 
			In it, $u$ and $v$ are not made happy within the bottom face, while it is trivial to satisfy that face locally.  
			This shows that  just answering whether a face can be satisfied locally is not enough information. 
		}
		\label{fig:sketch-badcase-convex-tree}
	\end{figure}
	
	\smallskip
	
	{
	\noindent\textbf{Case (i):} If $f$ has no \invertible child, then the vertices on face $f$ have no choice. Thus, we can immediately determine whether $\Pi(f,p_u,p_v)$ has a solution by applying \Cref{lemma:convexCycle,lemma:convexPath,lemma:convexPathsCollection} on the instance in $f$ to fulfill the parity constraints from $\Pi(f,p_u,p_v)$. 
	}
	\smallskip
	
	{
	\noindent\textbf{Case (ii):} If $f$ has at most two \invertible children, say with connector edges $wx$ and $yz$, we consider all at most sixteen choices for $p_w,p_x,p_y,p_z\in \{0,1\}$. 
	For each of them, we first check whether there is a solution in the children with these parity constraints. 
	(Note that we only need to consider four choices, because the number of unhappy vertices in any problem needs to be even for there to exist a solution.)
	Then, for each of these choices that is valid for the children, 
	we produce an instance for $f$ where we adjust the happiness status of $w,x,y,z$ so that the combination with $p_w,p_x,p_y,p_z$ gives the desired parities. 
	This instance can be checked using \Cref{lemma:convexCycle,lemma:convexPath,lemma:convexPathsCollection}.
	}
	\smallskip 
	
	{
	\noindent\textbf{Case (iii):} If $f$ has at least three \invertible children, we show that  $\Pi(f,p_u,p_v)$ always has a \hset if  $p_u,p_v$ are such that the total number of unhappy vertices is even; otherwise there cannot be a solution.}
	If $f$ has multiple \invertible children, 
	trying all parity combinations for the vertices could lead to exponentially many cases.  
	We are saved by the following observation: 
	when a face contains at least three \invertible edges, 
	then there is a choice for them that produces two non-consecutive vertices on $f$ that do not need to change their parities in $f$ to satisfy $\Pi(f,p_u,p_v)$. 
	\Cref{lemma:convexCycle} then ensures that any of the two options of parity constraints for $u$ and $v$ such that the total number of unhappy vertices in the problem is even, are satisfiable in $f$.

	For a face $f$, if $\Pi(f,p_u,p_v)$ admits a solution for more than one value of $p_u,p_v$ (there are at most two such problems that can be solvable), 
	$f$ and the connector edge $uv$ are marked \invertible and, if $f \ne f_0$, we proceed to a new face. For an example consider \cref{fig:convexTreeB}.
	If at any face none of the problems admits a solution we conclude that no \hset for $(G,R)$ exists. 
	If we reached the root $f_0$ and the problem $\Pi(f_0,p_u,p_v)$ where $p_u,p_v$ match the original parity constraints has a solution, we conclude that $(G,R)$ admits a \hset. 
	Note that using this algorithm it is easy to efficiently construct a \hset for $(G,R)$ if one exists.

	The correctness of the algorithm is based on the following fact. 
	Let $uv \in E$ be the connector edge for $f$.
	Then, $uv$ separates $G$ into two subgraphs, one of which contains 
	the vertices 
	in $f$ and 
	all the descendant faces. 
	Thus, $uv$ defines two subinstances of the problem that only overlap in $u$ and $v$.  
	There exists a \hset for $(G,R)$ if and only if those subinstances admit solutions that together satisfy the parity constraints of $u$ and $v$. 
	This implies that the algorithm checks all the possibilities, as it keeps all possible parities of $u$ and $v$ in valid solutions of the subinstance considered so far.
	
	It remains to show that the constructed solutions are valid. 
	Processing face $f$ requires to satisfy all the parity constraints from $R$ on all the vertices in $f$ except for $u$ and $v$. 
	The last check in the root face $f_0$ enforces the parity constraints from $R$ on all its vertices.
	Thus, for each vertex, the last face containing it enforces the correct parity.
	
	The time spent by the algorithm in each face is linear in its size. 
	Thus, given the convexly \hug $C$, $D(G)$ can be computed in linear time and the algorithm runs in $O(|V|)$ time.
\end{proof}

\section{Plane Geometric Paths}
\label{sec:paths}
In this section, we efficiently solve the decision problem for plane geometric paths.

\begin{theorem}
	\label{thm:paths}
	Let $G=(V,E)$ be a plane geometric path and let $R\subseteq V$. There exists an algorithm to decide whether $(G,R)$ admits a \hset in $O(|V|\log |V|)$ time.
\end{theorem}

To prove this result, we consider two cases. 
Either the path is \emph{universally happy}, i.e., there exist \hset{s} for any unhappy subset $R$ of even cardinality, or the path has a particular structure (called \emph{\pseudoconvex{ity}}) which allows us to use our algorithm for \Cref{thm:HUG}.

\subsection{Universal Happiness}

Universal happiness requires \Vis{G} to be connected: 
otherwise there exists no \hset when $R$ consist of an odd number of vertices from two different components of \Vis{G}; see \Cref{fig:zigzag}.

\begin{figure}[htb]
	\centering
	\begin{subfigure}{.45\textwidth}
		\centering
		\includegraphics[page=9]{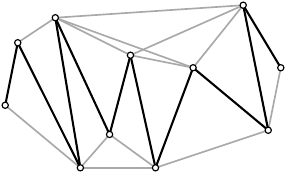}
		\caption{}
		\label{fig:zigzagA}
	\end{subfigure}
	\hfil
	\begin{subfigure}{.45\textwidth}
		\centering
		\includegraphics[page=10]{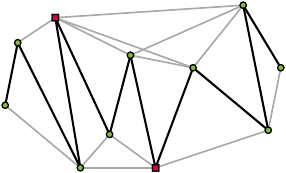}
		\caption{}
		\label{fig:zigzagB}
	\end{subfigure}
	\hfil
	
	\caption{An example in which \Vis{G} is not connected and an even set of two unhappy vertices (square) that does not admit a solution.}
	\label{fig:zigzag}
\end{figure}

In the next lemma we show that the existence of a plane spanning tree in \Vis{G} is sufficient as long as $|R|$ is even. 
The proof uses standard arguments (for example, in the context of T-joins, and also similar to Theorem 3 in \cite{maxcut}).

\begin{lemma}
\label{lem:Sufficient1}
	Let $G=(V,E)$ be a plane geometric graph such that \Vis{G} has a plane spanning tree $T$. 
	Then for all $R\subseteq V$ with $|R|$ even $(G,R)$ admits a \hset.
\end{lemma}

\begin{proof}
	Consider a fixed but arbitrary $R\subseteq V$ with $|R|$ even. 
	We partition $R$ into pairs. 
	For each such pair $(u,v)$, we consider the unique $uv$-path $P_{uv}$ in $T$. 
	Augmenting $G$ by $P_{uv}$ satisfies the parity constraints for $u$ and $v$ without altering the happiness status of the other vertices.  
	To construct a \hset we can iterate over the paths $P_{uv}$ swapping their edges. 
	More formally, 
	let $H^m$ denote the edge multi-set of the union of all these paths $P_{uv}$. 
	Deleting two copies of an edge maintains the parities of the degrees of all vertices. 
	Consequently, we can define a \hset $H$ for $(G,R)$ as the subset of $H^m$ consisting on the edges of odd multiplicity; see \Cref{fig:spanningTree} for an illustration.
\end{proof}

\begin{figure}[htb]
	\centering
	\begin{subfigure}{.3\textwidth}
		\centering
		\includegraphics[page=11]{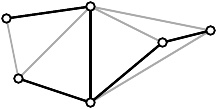}
		\caption{}
		\label{fig:spanningTreeA}
	\end{subfigure}\hfil
	\begin{subfigure}{.3\textwidth}
		\centering
		\includegraphics[page=12]{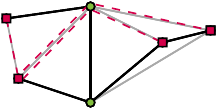}
		\caption{}
		\label{fig:spanningTreeB}
	\end{subfigure}\hfil
	\begin{subfigure}{.3\textwidth}
		\centering
		\includegraphics[page=13]{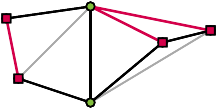}
		\caption{}
		\label{fig:spanningTreeC}
	\end{subfigure}\hfil
	\caption{Illustration for the proof of \Cref{lem:Sufficient1}. {(a) A plane geometric graph $G$ (in black) and a plane spanning tree $T$ of \Vis{G} (in gray). (b) A set $R$ and paths (in dashed red) between (an arbitrary partition of $R$ into) pairs of $R$ in $T$. (c) The symmetric difference of the paths yields a \hset (in red).}}
	\label{fig:spanningTree}
\end{figure}


Garc\'ia~\etal\cite{GARCIA2014} give a characterization of the existence of a plane spanning tree in the visibility graph for the case of a path. 
To use it, we introduce some notation. 
Let $G = (V,E)$ be a plane geometric path, and let $q_0, \dots, q_{k-1}$ be the vertices on the convex hull of $V$ ordered clockwise. 
Adding an edge $e = q_{i}q_{i+1}\in \Vis{G}$ to $G$ creates a unique face (indices are considered modulo~$k$). 
Its boundary is a geometric graph that we call \emph{pocket}: 
it consists of the edge $e$ together with the subpath of $G$ connecting the endpoints of $e$.  
We say that a pocket $P = v_{0},v_{1}, \dots, v_{s}$ with $v_{0} = q_{i}, v_{s} = q_{i+1}$ for some $i$ is \emph{\pseudoconvex} if 
it does not contain an endpoint of $G$ in its interior and it is either convex or  every reflex vertex  $v_{j} \in P$  
satisfies the condition that the rays obtained by extending the edges $v_{j-1}v_{j}$ and $v_{j+1}v_{j}$ (indices considered modulo~$s$) intersect $P$ for the first time at $e = q_{i}q_{i+1}$. 
Note that if $G$ is \pseudoconvex, its two endpoints lie on the convex hull.
\Cref{fig:qConvexA} illustrates a path and its pockets. 
If every pocket of $G$ is \pseudoconvex, we say that $G$ is \pseudoconvex.

\begin{figure}[hbt]
	\centering
		\includegraphics[page=18]{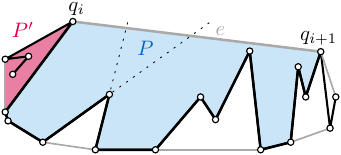}
	\caption{
		Example of a path and its pockets.  While pocket $P$ (with blue interior) is  \pseudoconvex, pocket $P'$ (with red interior) is not.}
	\label{fig:qConvexA}
\end{figure}

We exploit the following fact which has been showed by Garc\'ia~\etal\cite{GARCIA2014} using a slightly different language.

\begin{lemma}[Theorem 3 and Remark 2 in~\cite{GARCIA2014}]\label{thm:spanningTree}
	Let $G=(V,E)$ be a plane geometric path. \Vis{G} has a plane spanning tree if and only if $G$ is not \pseudoconvex.
\end{lemma}

Using \Cref{thm:spanningTree} and our tools of the proof for \Cref{thm:paths}, we show in \Cref{subsec::univHappinessPaths} that the sufficient condition in \Cref{lem:Sufficient1} is also necessary.

\begin{theorem}
\label{thm:univHappinessPaths}
	Let $G=(V,E)$ be a plane geometric path. This geometric graph $G$ is universally happy if and only if \Vis{G} has a plane spanning tree, or equivalently, $G$ is not pseudo-convex.
\end{theorem}

\subsection{Tight Hulls and Proof of \Cref{thm:paths}}
Our main tools to prove \Cref{thm:paths} are \Cref{thm:HUG} and \Cref{thm:spanningTree,lem:Sufficient1}. The last ingredient is to bridge the gap between \pseudoconvex{ity} and convexity.

{Let $G = (V,E)$ be a plane geometric \pseudoconvex path and consider a} (\pseudoconvex) pocket $P$ induced by some edge $e = q_{i}q_{i+1} \in \Vis{G}$.  Let $r_{1}, \dots, r_{s}$ be the reflex vertices of $P$ ordered counterclockwise. 
The \emph{tight hull} of $P$ consists of the edges $q_{i}r_{1},r_{1}r_{2}, \dots, r_{s-1}r_{s}, r_{s}q_{i+1}$; see \Cref{fig:qConvexB} for an illustration. 
Each edge on the tight hull of $P$ defines a unique face, creating a \emph{subpocket} of $P$. 
Garc\'ia~\etal\cite{GARCIA2014} show that this tight hull always exists and each subpocket is convex (Lemma 4.i in~\cite{GARCIA2014}). Replacing the convex hull edge of every pocket of $G$ by the tight hull of the pocket, one gets the tight hull of $G$. 

\begin{figure}[hbt]
		\centering
		\includegraphics[page=17]{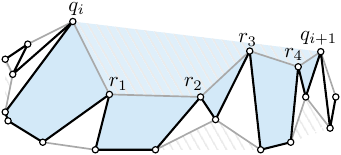}
	\caption{Example of a \pseudoconvex path and the tight hull of a pocket (with blue interior), its subpockets (with darker blue interior), and the tight hull of $G$ forming a \hug (in gray).}
	\label{fig:qConvexB}
\end{figure}

\begin{lemma}
	\label{lemma:hug}
	Let $G=(V,E)$ be a plane geometric \pseudoconvex path. Then the tight hull $T$ of $G$ is a convexly \hug.
\end{lemma}
\begin{proof}
	 By definition, the edges of $G$ lie in the (closed) interior of $T$. 
	 We now show that~$T$ contains all vertices. 
	 Suppose there exists a vertex $v$ in the strict interior of $T$. 
	 Then $v$ is contained in the interior of some pocket of $G$, guaranteeing that the pocket and thus $G$ are not \pseudoconvex.
	 We note that $T$ is a cycle: 
	 by construction it is connected and each vertex has degree 2. 
	 Finally, the convexity of all bounded faces of $G\cup T$ follows from the convexity of the subpockets.
\end{proof}

Looking for a \hset of a \pseudoconvex path, we now show that we may restrict our attention to the tight hull $T$, more precisely, to $\VisT{G}$, the restriction of \Vis{G} to the tight hull and its interior.
\begin{lemma}
\label{lemma:tightHull}
	Let $G=(V,E)$ be a plane geometric \pseudoconvex path, $T$ its tight hull, and let $R\subseteq V$. 
	If $(G,R)$ allows a \hset $H$ then there also exists a \hset $H'$ within $\VisT{G}$. 
\end{lemma}

\begin{proof}	
	We distinguish two types of \emph{unwanted} edges in $H$: 
	(i) edges crossing the tight hull
	and (ii) edges in the exterior of the tight hull (not on its boundary).

Let $uv$ be an unwanted edge of Type (i) in a pocket $P$ of $G$.  
	Recall that subpockets are convex. 
	Since $uv$ is part of a \hset, it can cross at most two tight hull edges.
	Let $e$ be the edge of the tight hull crossed by $uv$ such that the subpocket $P^\circ$ of $e$ includes $u$. 
By convexity of $P^\circ$, 
vertex $u$ is convex in $P^\circ$ (and in $P$); the other vertex $v$ lies outside $P^\circ$. For an illustration see \cref{fig:tightHullA}.

	\begin{figure}[htb]
		\centering
		\begin{subfigure}{.45\textwidth}
			\centering
			\includegraphics[page=20]{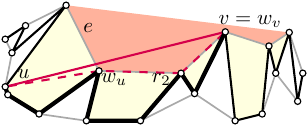}
			\caption{}
			\label{fig:tightHullA}
		\end{subfigure}\hfill
		\begin{subfigure}{.45\textwidth}
			\centering
			\includegraphics[page=21]{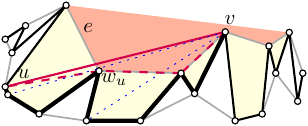}
			\caption{}
			\label{fig:tightHullB}
		\end{subfigure}
		\caption{
			Illustration for the proof of \Cref{lemma:tightHull}.
			(a) An edge $uv$ of Type (i) crossing the tight hull (in solid red) and the edges replacing it (in dashed red). \linebreak(b) {Edges from $H$ crossing replacement edges would be of Type (i) and define a shorter subpath in $G$, contradicting the choice of $uv$.} }
		\label{fig:tightHull}
	\end{figure}

	Among all such edges of {Type (i)}, we consider the edge $uv$ that minimizes the number of vertices on the subpath $\pi$ of $G$ between $u$ and $v$.
	Let $w_u$  be the first reflex vertex (of $P$) encountered when walking along $\pi$ from $u$ towards $v$, and let $w_v$ be the first reflex vertex (of $P$) encountered when walking along $\pi$ from $v$ towards $u$. 
	If $v$ itself is a reflex vertex, then let $w_v = v$.
	Note that, by convex{ity} of the subpockets, we have $uw_u\in\Vis{G}$ and $vw_v\in\Vis{G}$. 
	
	Let $w_u = r_1, r_2, \dots, r_{\ell} = w_v$ be the 
	reflex vertices of $P$ belonging to $\pi$, see also \cref{fig:tightHullA}.
	We aim to replace the \hset $H$ by $H \triangle \{uw_u, r_{1}r_{2}, \dots, r_{\ell-1}r_\ell,w_vv,vu\}$, where $\triangle$ denotes the symmetric difference. We show that none of the replacement edges intersect an edge of~$H$.
	
	To this end, we show that by the choice of~$uv$, the edges $uw_u$ 
	 and $vw_v$ do not intersect any edge of $H$. 
	 Any edge of $H$ crossing $uw_u$ but not $uv$ has to have one endpoint $x \in \pi[u, w_u]$ distinct from both $u$ and $w_u$, and one endpoint $y$ not in 	$\pi[u,w_u]$. 
	 If $y$ belongs to $\pi[w_u, v]$, then $xy$ is of {of Type~(i)} and thus contradicts the minimality of $uv$. For an illustration, see  \cref{fig:tightHullB}.
	 If $y$ is outside of $\pi$, then $xy$ crosses $uv$, which is not possible since both edges were part of the \hset $H$. 
	 Analogously, we have $vw_v\in\Vis{G}$.

	The situation is similar when considering a replacement edge  $r_{i}r_{i+1}$.
	Suppose that edge $h \in H$ intersects $r_{i}r_{i+1}$.
	Then $h$ is of {Type~(i)}. If both vertices of $h$ belong to $\pi$ then $h$ contradicts the minimality of $uv$. Otherwise, $h$ has exactly one endpoint not belonging to $\pi$, but then $h$ 
	must cross the closed curve $uv \cup \pi$, contradicting the properties of $H$.

	Thus, we replace the \hset $H$ by $H \triangle \{uw_u, r_{1}r_{2}, \dots, r_{\ell-1}r_\ell,w_vv,vu\}$ 
	and obtain a  non-crossing edge set.
	Because $\{uw_u, r_{1}r_{2}, \dots, r_{\ell-1}r_\ell,w_vv,vu\}$ forms a cycle, all degree parities are maintained. 
	Moreover, $H$ contains fewer unwanted edges of Type~(i). 
	Repeating this procedure, we obtain a \hset $H$ without unwanted edges of Type~(i).

	Suppose now that $H$ contains an unwanted edge of Type~(ii) between $u$ and $v$ in a pocket $P$; see \cref{fig:tightHullC}. 
	\begin{figure}[htb]
		\centering
			\includegraphics[page=22]{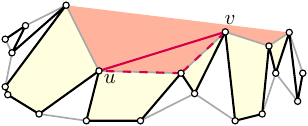}
		\caption{Illustration for the proof of \Cref{lemma:tightHull}. An edge $uv$ of Type (ii) in the interior of the tight hull. }
		
		\label{fig:tightHullC}
	\end{figure} 
	
	Let $\pi$ denote the path on the tight hull between $u$ and $v$ that lies inside $P$. 
	Since $uv$ is not on the tight hull, $\pi$ must contain other vertices apart from $u$ and $v$.
	The absence of Type (i) edges guarantees that $\pi$ does not cross any edge from $H$. 
	Thus, since $\pi \cup \{uv\}$ forms a cycle, replacing $H$ by $H \triangle (\pi \cup \{uv\})$ results in a \hset with fewer edges of Type~(ii). 
	Repeating this procedure, yields the lemma.
\end{proof}

\begin{proof}[Proof of \Cref{thm:paths}]
 First we check whether $G$ is \pseudoconvex: {Using existing ray shooting techniques in polygons~\cite{chazelle1994ray,GuibasHS91}, the first intersection point on a ray can be computed in $O(\log|V|)$ time. For $O(|V|)$ reflex vertices, we obtain $O(|V|\log |V|)$ time, which also bounds the preprocessing time required.} 
 If $G$ is not \pseudoconvex, then \Cref{thm:spanningTree} guarantees the existence of a plane spanning tree in \Vis{G} and $G$ is universally happy by \cref{lem:Sufficient1}.
  If $G$ is  \pseudoconvex, the tight hull $T$ {can be computed in linear time by considering the vertices in the order given by the path. By \cref{lemma:hug},} 
  $T$ is a convexly \hug  and \cref{lemma:tightHull} guarantees a \hset in $\VisHug{T}{G}$  if there exists a \hset in  $\Vis{G}$.
 Hence, we may use the linear-time algorithm given in the proof of \Cref{thm:HUG} to decide whether there exists a \hset.
\end{proof}

\subsection{Proof of \Cref{thm:univHappinessPaths}}
\label{subsec::univHappinessPaths}

We are now ready to show that a plane geometric path is universally happy if and only if its visibility graph has a plane spanning tree.

\begin{proof} [Proof of \Cref{thm:univHappinessPaths}]
	By \Cref{lem:Sufficient1}, we know that the existence of a plane spanning tree in \Vis{G} implies universal happiness of $G$. It remains to show that it is a necessary condition.
 	
	Suppose \Vis{G} does not contain a plane spanning tree. 
	Then, by \Cref{thm:spanningTree}, $G$ is \pseudoconvex. 
	{By \cref{lemma:hug}, the tight hull $T$ of $G$ is a \hug and}
	by \Cref{lemma:tightHull}, we may restrict our attention to edges in $\VisT{G}$. 
	We now construct a 
	set $R$ without a \hset by considering the weak dual tree $D(G)$ of the outerplane graph $G\cup T$. 
	Root $D(G)$ at an arbitrary leaf face $f_{0}$. Note that $D(G)$ is a path on vertices $f_{0},f_{1},\dots, f_{k}$. For an illustration, consider \Cref{fig:pathUnhappy}.
	In particular, we consider the faces in reverse order, beginning from the other leaf face $f_k$ (as the algorithm of \Cref{thm:HUG} would do in order to compute a \hset or decide that none exist).

	\begin{figure}[htb]
		\centering
		\begin{subfigure}{.45\textwidth}
			\centering
			\includegraphics[page=23]{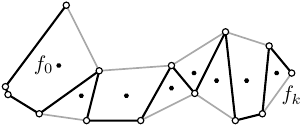}
			\caption{}
			\label{fig:pathUnhappyA}
		\end{subfigure}\hfill
	\begin{subfigure}{.45\textwidth}
		\centering
		\includegraphics[page=24]{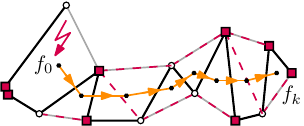}
		\caption{}
		\label{fig:pathUnhappyB}
	\end{subfigure}\hfill
		\caption{Illustration for the proof of \Cref{thm:univHappinessPaths}. (a) A \pseudoconvex path $G$ and its weak dual tree $D(G)$ of $G\cup T$. (b) A set of unhappy vertices $R$ that has no \hset.}
		\label{fig:pathUnhappy}
	\end{figure}

	As before, we denote the dual edge of $f_{i}f_{i-1}$ as the {connector edge} of face~$f_{i}$. 
	For a face $f \neq f_{0}$, let $V_{f}$ be the set of vertices on its boundary and $uv$ its connector edge. 
	By renaming, we may assume that $u$ is incident to the unique edge shared by $f$ and the tight hull. 
	For $f=f_0$, let $u$ denote the endpoint of $G$ in $f_0$ and $v$ its neighbor in $G$.
	For $f=f_k$, we define $R:=V_{f}$ if $|V_{f}|$ is even and $R:=V_{f} \setminus \{u\}$ if $|V_{f}|$ is odd.
	Clearly, the vertices $V_f\setminus\{u,v\}$ can only be satisfied by edges inside $f$. 
	By \Cref{lemma:convexPath}, there exists a \hset inside $f$ satisfying $V_f\setminus\{u,v\}$ only if $v$ serves as the happy inner vertex. 
	More precisely, the unique way to satisfy $V_{f} \setminus \{u,v\}$ inside $f$ is by adding the star with center $u$ and leaves $V_{f} \setminus \{u,v\}$. 
	To see that, note that any plane triangulation of $f$ has at least two vertices of degree 2. The only choices for these vertices are $v$ and the neighbor of $u$ on the tight hull of $f$.
	It follows that the connector vertices $u$ and $v$ are thus guaranteed to be unhappy after inserting edges in $f_{k}$ to satisfy the vertices of $V_{f_k} \setminus \{u,v\}$.
	
	Processing the faces in the order $f_k,f_{k-1}, f_{k-2},\dots, f_0$, we now add vertices of $G$ to $R$.
	To this end, let $U_f:=\bigcup_{f' \text{ after $f$ in $D(G)$}} V_{f'}$ denote the set of all previously considered vertices. 
	For a face $f$, we add $V_{f}\setminus U_f$ to $R$ if $|V_{f}|$ is even and $V_{f} \setminus \{u\}\cup U_f$ otherwise. 
	This maintains the invariant that $u$ and $v$ of a face $f$ are unhappy in any \hset satisfying $U_f\setminus \{u,v\}$ in the later faces. 
	It follows that all vertices of $V_f$ are unhappy and thus for any face $f\neq f_0$ in any \hset, all vertices in $V_{f} \setminus \{v\}$ must be connected to $u$. 
	Finally, in $f_0$ every vertex in $V_f\setminus \{u\}$ is unhappy. 
	Consequently, no inner vertex is happy, and by \Cref{lemma:convexPath}, there exists no \hset.
\end{proof}

\section{Conclusion}

We showed how to decide in polynomial time whether a plane geometric path can be augmented (with straight-line edges and without creating crossings) to meet some given parity constraints of the vertices. The main question that remains is whether an analogous result exists for plane geometric trees. More restricted versions of this problem are also open, in particular when set of unhappy vertices consists of all the vertices or all the odd-degree vertices.


\bibliography{bib}

\begin{thebibliography}{10}

\bibitem{ABELLANAS2008}
Manuel Abellanas, Alfredo~Garc{\'{\i}}a Olaverri, Ferran Hurtado, Javier Tejel,
  and Jorge Urrutia.
\newblock Augmenting the connectivity of geometric graphs.
\newblock {\em Comput. Geom.}, 40(3):220--230, 2008.
\newblock \href {https://doi.org/10.1016/j.comgeo.2007.09.001}
  {\path{doi:10.1016/j.comgeo.2007.09.001}}.

\bibitem{AdriaensG22}
Florian Adriaens and Aristides Gionis.
\newblock Diameter minimization by shortcutting with degree constraints.
\newblock In Xingquan Zhu, Sanjay Ranka, My~T. Thai, Takashi Washio, and
  Xindong Wu, editors, {\em Proc. {IEEE} International Conference on Data
  Mining ({ICDM} 2022)}, pages 843--848. {IEEE}, 2022.
\newblock \href {https://doi.org/10.1109/ICDM54844.2022.00095}
  {\path{doi:10.1109/ICDM54844.2022.00095}}.

\bibitem{Aichholzer2014}
Oswin Aichholzer, Thomas Hackl, Michael Hoffmann, Alexander Pilz, G{\"{u}}nter
  Rote, Bettina Speckmann, and Birgit Vogtenhuber.
\newblock Plane graphs with parity constraints.
\newblock {\em Graphs Comb.}, 30(1):47--69, 2014.
\newblock \href {https://doi.org/10.1007/s00373-012-1247-y}
  {\path{doi:10.1007/s00373-012-1247-y}}.

\bibitem{AKITAYA2019}
Hugo~A. Akitaya, Rajasekhar Inkulu, Torrie~L. Nichols, Diane~L. Souvaine,
  Csaba~D. T{\'o}th, and Charles~R. Winston.
\newblock Minimum weight connectivity augmentation for planar straight-line
  graphs.
\newblock {\em Theor. Comput. Sci.}, 789:50--63, 2019.
\newblock \href {https://doi.org/10.1016/j.tcs.2018.05.031}
  {\path{doi:10.1016/j.tcs.2018.05.031}}.

\bibitem{AlJubeh2011}
Marwan Al-Jubeh, Mashhood Ishaque, Krist{\'o}f R{\'e}dei, Diane~L. Souvaine,
  Csaba~D. T{\'o}th, and Pavel Valtr.
\newblock Augmenting the edge connectivity of planar straight line graphs to
  three.
\newblock {\em Algorithmica}, 61(4):971, 2011.
\newblock \href {https://doi.org/10.1007/s00453-011-9551-0}
  {\path{doi:10.1007/s00453-011-9551-0}}.

\bibitem{AlonGR00}
Noga Alon, Andr{\'{a}}s Gy{\'{a}}rf{\'{a}}s, and Mikl{\'{o}}s Ruszink{\'{o}}.
\newblock Decreasing the diameter of bounded degree graphs.
\newblock {\em J. Graph Theory}, 35(3):161--172, 2000.
\newblock \href {https://doi.org/10/b39gzn} {\path{doi:10/b39gzn}}.

\bibitem{Alvarez2015}
V{\'{\i}}ctor {\'{A}}lvarez.
\newblock Parity-constrained triangulations with {S}teiner points.
\newblock {\em Graphs Comb.}, 31(1):35--57, 2015.
\newblock \href {https://doi.org/10.1007/s00373-013-1389-6}
  {\path{doi:10.1007/s00373-013-1389-6}}.

\bibitem{BenczurK00}
Andr{\'{a}}s~A. Bencz{\'{u}}r and David~R. Karger.
\newblock Augmenting undirected edge connectivity in {${\tilde O}(n^2)$} time.
\newblock {\em J. Algorithms}, 37(1):2--36, 2000.
\newblock \href {https://doi.org/10.1006/jagm.2000.1093}
  {\path{doi:10.1006/jagm.2000.1093}}.

\bibitem{Catanaetal}
Juan~C. Catana, Alfredo~Garc{\'{\i}}a Olaverri, Javier Tejel, and Jorge
  Urrutia.
\newblock Plane augmentation of plane graphs to meet parity constraints.
\newblock {\em Appl. Math. Comput.}, 386:125513, 2020.
\newblock \href {https://doi.org/10.1016/j.amc.2020.125513}
  {\path{doi:10.1016/j.amc.2020.125513}}.

\bibitem{CenLP22}
Ruoxu Cen, Jason Li, and Debmalya Panigrahi.
\newblock Augmenting edge connectivity via isolating cuts.
\newblock In {\em Proc. 2022 {ACM-SIAM} Symposium on Discrete Algorithms
  ({SODA} 2022)}, pages 3237--3252. {SIAM}, 2022.
\newblock \href {https://doi.org/10.1137/1.9781611977073.127}
  {\path{doi:10.1137/1.9781611977073.127}}.

\bibitem{chazelle1994ray}
Bernard Chazelle, Herbert Edelsbrunner, Michelangelo Grigni, Leonidas Guibas,
  John Hershberger, Micha Sharir, and Jack Snoeyink.
\newblock Ray shooting in polygons using geodesic triangulations.
\newblock {\em Algorithmica}, 12(1):54--68, 1994.
\newblock \href {https://doi.org/10.1007/BF01377183}
  {\path{doi:10.1007/BF01377183}}.

\bibitem{ChungG84}
Fan R.~K. Chung and Michael~R. Garey.
\newblock Diameter bounds for altered graphs.
\newblock {\em J. Graph Theory}, 8(4):511--534, 1984.
\newblock \href {https://doi.org/10.1002/jgt.3190080408}
  {\path{doi:10.1002/jgt.3190080408}}.

\bibitem{DABROWSKI2016213}
Konrad~K. Dabrowski, Petr~A. Golovach, Pim {van 't Hof}, and Dani{\"e}l
  Paulusma.
\newblock Editing to {E}ulerian graphs.
\newblock {\em Journal of Computer and System Sciences}, 82(2):213--228, 2016.
\newblock \href {https://doi.org/10.1016/j.jcss.2015.10.003}
  {\path{doi:10.1016/j.jcss.2015.10.003}}.

\bibitem{EdmondsJ73}
Jack Edmonds and Ellis~L. Johnson.
\newblock Matching, euler tours and the chinese postman.
\newblock {\em Mathematical Programing}, 5(1):88--124, 1973.
\newblock \href {https://doi.org/10.1007/BF01580113}
  {\path{doi:10.1007/BF01580113}}.

\bibitem{Tarjan}
Kapali~P. Eswaran and R.~Endre Tarjan.
\newblock Augmentation problems.
\newblock {\em {SIAM} J. Comput.}, 5(4):653--665, 1976.
\newblock \href {https://doi.org/10.1137/0205044} {\path{doi:10.1137/0205044}}.

\bibitem{Frank92}
Andr{\'{a}}s Frank.
\newblock Augmenting graphs to meet edge-connectivity requirements.
\newblock {\em {SIAM} J. Discret. Math.}, 5(1):25--53, 1992.
\newblock \href {https://doi.org/10.1137/0405003} {\path{doi:10.1137/0405003}}.

\bibitem{FratiGGM15}
Fabrizio Frati, Serge Gaspers, Joachim Gudmundsson, and Luke Mathieson.
\newblock Augmenting graphs to minimize the diameter.
\newblock {\em Algorithmica}, 72(4):995--1010, 2015.
\newblock \href {https://doi.org/10.1007/s00453-014-9886-4}
  {\path{doi:10.1007/s00453-014-9886-4}}.

\bibitem{GARCIA2014}
Alfredo Garc{\'\i}a, Clemens Huemer, Ferran Hurtado, and Javier Tejel.
\newblock Compatible spanning trees.
\newblock {\em Comput. Geom.}, 47(5):563--584, 2014.
\newblock \href {https://doi.org/10.1016/j.comgeo.2013.12.009}
  {\path{doi:10.1016/j.comgeo.2013.12.009}}.

\bibitem{Garcia2015}
Alfredo Garc{\'\i}a, Ferran Hurtado, Matias Korman, In{\^e}s Matos, Maria
  Saumell, Rodrigo~I. Silveira, Javier Tejel, and Csaba~D. T{\'o}th.
\newblock Geometric biplane graphs {II}: Graph augmentation.
\newblock {\em Graphs Comb.}, 31(2):427--452, 2015.
\newblock \href {https://doi.org/10.1007/s00373-015-1547-0}
  {\path{doi:10.1007/s00373-015-1547-0}}.

\bibitem{GuibasHS91}
Leonidas~J. Guibas, John Hershberger, and Jack Snoeyink.
\newblock Compact interval trees: a data structure for convex hulls.
\newblock {\em Int. J. Comput. Geom. Appl.}, 1(1):1--22, 1991.
\newblock \href {https://doi.org/10.1142/S0218195991000025}
  {\path{doi:10.1142/S0218195991000025}}.

\bibitem{maxcut}
F.~Hadlock.
\newblock Finding a maximum cut of a planar graph in polynomial time.
\newblock {\em SIAM Journal on Computing}, 4(3):221--225, 1975.
\newblock \href {https://doi.org/10.1137/0204019} {\path{doi:10.1137/0204019}}.

\bibitem{KantBodl1991}
Goos Kant and Hans~L. Bodlaender.
\newblock Planar graph augmentation problems.
\newblock In {\em Proc. 2nd Workshop on Algorithms and Data Structures ({WADS}
  1991)}, volume 519 of {\em LNCS}, pages 286--298. Springer, 1991.
\newblock \href {https://doi.org/10.1007/BFb0028270}
  {\path{doi:10.1007/BFb0028270}}.

\bibitem{KRANAKIS2010}
Evangelos Kranakis, Danny Krizanc, Oscar Morales{-}Ponce, and Ladislav Stacho.
\newblock Bounded length, 2-edge augmentation of geometric planar graphs.
\newblock {\em Discret. Math. Algorithms Appl.}, 4(3), 2012.
\newblock \href {https://doi.org/10.1142/S179383091250036X}
  {\path{doi:10.1142/S179383091250036X}}.

\bibitem{ChinesePostman1}
Mei-ko Kwan.
\newblock Graphic programming using odd or even points.
\newblock {\em Chinese Math}, 1:273--277, 1960.

\bibitem{Melkman87}
Avraham~A. Melkman.
\newblock On-line construction of the convex hull of a simple polyline.
\newblock {\em Inf. Process. Lett.}, 25(1):11--12, 1987.
\newblock \href {https://doi.org/10.1016/0020-0190(87)90086-X}
  {\path{doi:10.1016/0020-0190(87)90086-X}}.

\bibitem{nash-williams_1960}
C.~ST. J.~A. Nash-Williams.
\newblock On orientations, connectivity and odd-vertex-pairings in finite
  graphs.
\newblock {\em Canadian Journal of Mathematics}, 12:555–567, 1960.
\newblock \href {https://doi.org/10.4153/CJM-1960-049-6}
  {\path{doi:10.4153/CJM-1960-049-6}}.

\bibitem{PilzRS020}
Alexander Pilz, Jonathan Rollin, Lena Schlipf, and Andr{\'{e}} Schulz.
\newblock Augmenting geometric graphs with matchings.
\newblock In {\em Proc. 28th International Symposium on Graph Drawing and
  Network Visualization ({GD} 2020)}, volume 12590 of {\em LNCS}, pages
  490--504. Springer, 2020.
\newblock \href {https://doi.org/10.1007/978-3-030-68766-3\_38}
  {\path{doi:10.1007/978-3-030-68766-3\_38}}.

\bibitem{Plesnik1976}
J\'an Plesn{\'\i}k.
\newblock Minimum block containing a given graph.
\newblock {\em Arch. Math.}, 27(6):668--672, 1976.

\bibitem{RutterWolff2012}
Ignaz Rutter and Alexander Wolff.
\newblock Augmenting the connectivity of planar and geometric graphs.
\newblock {\em J. Graph Algorithms Appl.}, 16(2):599--628, 2012.
\newblock \href {https://doi.org/10.7155/jgaa.00275}
  {\path{doi:10.7155/jgaa.00275}}.

\bibitem{TOTH2012}
Csaba~D. T{\'o}th.
\newblock Connectivity augmentation in planar straight line graphs.
\newblock {\em Eur. J. Comb.}, 33(3):408--425, 2012.
\newblock \href {https://doi.org/10.1016/j.ejc.2011.09.002}
  {\path{doi:10.1016/j.ejc.2011.09.002}}.

\end{thebibliography}

\end{document}